\documentclass[twoside]{article}

\usepackage{amsmath,amsthm,amssymb}

\usepackage{newtxtext,newtxmath}

\usepackage[text={12.5cm,19cm},centering,paperwidth=17cm,paperheight=24cm]{geometry}

\usepackage[fontsize=10.4pt]{scrextend}

\def\titlerunning#1{\gdef\titrun{#1}}

\makeatletter
\def\author#1{\gdef\autrun{\def\and{\unskip, }#1}\gdef\@author{#1}}
\def\address#1{{\def\and{\\\hspace*{15.6pt}}\renewcommand{\thefootnote}{}\footnote{#1}}\markboth{\autrun}{\titrun}}
\makeatother

\def\email#1{email: \href{mailto:#1}{#1} }
\def\subjclass#1{\par\bigskip\noindent\textbf{Mathematics Subject Classification 2020.} #1}
\def\keywords#1{\par\smallskip\noindent\textbf{Keywords.} #1}

\newenvironment{dedication}{\itshape\center}{\par\medskip}
\newenvironment{acknowledgments}{\bigskip\small\noindent\textit{Acknowledgments.}}{\par}

\newtheorem{thm}{Theorem}[section]

\newtheorem{lem}[thm]{Lemma}

\newtheorem{proposition}[thm]{Proposition}
\newtheorem{conjecture}[thm]{Conjecture}

\theoremstyle{definition}

\newtheorem*{rem}{Remark}
\newtheorem*{remarks}{Remarks}

\numberwithin{equation}{section}

\usepackage[hyperfootnotes=false,colorlinks=true,allcolors=blue]{hyperref}

\begin{document}

\titlerunning{Bloch-Torrey equation}

\title{\textbf{On the spectral properties of the Bloch-Torrey equation in infinite periodically perforated domains}}

\author{Denis S. Grebenkov \and Nicolas Moutal  \and Bernard Helffer}

\date{}

\maketitle

\address{D.S. Grebenkov and N. Moutal : Laboratoire de Physique de la Mati\`{e}re Condens\'{e}e (UMR 7643), 
CNRS -- Ecole Polytechnique, IP Paris, 91128 Palaiseau, France; \email{denis.grebenkov@polytechnique.edu; nicolas.moutal@polytechnique.edu} \and B. Helffer:  Laboratoire de Math\'ematiques Jean Leray, \
CNRS and Universit\'e de Nantes, Nantes Cedex, France; \email{Bernard.Helffer@univ-nantes.fr}}

\begin{dedication}
To Ari Laptev on the occasion  of his 70th birthday.
\end{dedication}

\begin{abstract}
We investigate spectral and asymptotic properties of the particular
Schr\"o\-dinger operator (also known as the Bloch-Torrey operator),
$-\Delta + i g x$, in infinite periodically perforated domains of
$\mathbb R^d$.  We consider Dirichlet realizations 
of this operator and formalize a numerical approach proposed in
\cite{MMGghm} for studying such operators.  In particular, we discuss the
existence of the spectrum of this operator and its asymptotic behavior
as $g\to \infty$. 
\end{abstract}

\section{Introduction}

The aim of this paper is to formalize on the mathematical side a
numerical approach proposed in \cite{MMGghm} for analyzing the
Bloch-Torrey equation in infinite periodically perforated domains.
More precisely, we consider the  Dirichlet realization 
of the Bloch-Torrey operator
\begin{equation}
 B: = -\Delta + i g x\,,
\end{equation}
denoted respectively by $\mathcal B^D$ where $\Delta$ is the Laplace
operator, $x$ is one of the Cartesian coordinates, and $g$ is a real
nonzero parameter.   The Neumann and Robin realizations, which
are important in physical applications
\cite{Grebenkov07ghm,Gr1ghm,Grebenkov18ghm, MMGghm}, could also be treated by
the same techniques  but the details of the proof are omitted in the present paper.  The
major novelty of this work (as compared to former studies in
\cite{GHHghm,GHghm,AGH0ghm,AGHghm}) is that we consider the
Bloch-Torrey operator in a periodically perforated domain ($d\geq 2$),
\begin{equation}\label{defd1}
 \Omega = \mathbb R^d \setminus \{\cup_{\gamma \in \mathbb Z^d}
H_\gamma \}\,,
\end{equation}
 where 
 \begin{equation}\label{defd2}
 H_{\gamma} = \{ \mathbf{x} \in \mathbb R^d~:~
\mathbf{x} - \gamma \in H_0\} 
\end{equation}
and $H_0 \subset (-1/2,1/2)^d$ is a
domain  with a smooth boundary.\\

 A typical example is the case when $H_{\gamma} = B_2(\gamma,r)$ where  $B_2(\gamma,r)$ is the disk of radius $r <
\frac12$ centered at $\gamma$.  One of the major
difficulties in the definition and study of such non-self-adjoint
operators is that the potential $igx$ is not periodic, unbounded and
changing sign.

The Bloch-Torrey operator describes the diffusion-precession of
spin-bearing particles in nuclear magnetic resonance experiments and
helps in understanding the intricate relation between the geometric
structure of a studied sample (domain) and the measured signal
\cite{Grebenkov07ghm,Gr1ghm,Grebenkov18ghm}.  The spectral properties
of this operator and its asymptotic behavior play thus a crucial role
in this analysis.  \\

The paper is organized as follows.  In Sec. \ref{sec:BT_def}, we
provide a rigorous definition of the considered realization of the
Bloch-Torrey operator and describe its basic properties.  In
Sec. \ref{sec:Floquet}, we use the periodicity in $y$-direction via
the Floquet theory to reduce the operator on the planar perforated
domain to a family (indexed by a Floquet parameter) of operators on
the infinite perforated cylinder.  In particular, we formulate here
two conjectures about their spectral properties.  Section
\ref{sec:x-dir} analyzes the role of the pseudo-periodicity in
$x$-direction, which is specific to the Bloch-Torrey operator.  In
Sec. \ref{sec:quasimodes}, we formulate in an asymptotic regime the
main results concerning the non-emptiness of the spectrum and the
asymptotic properties.  Finally, Sec. \ref{sec:conclusion} concludes
the paper and discusses the extensions in the proofs to more general
settings.

\section{The Bloch-Torrey operator in the perforated whole plane}
\label{sec:BT_def}

We start from the Dirichlet realization of the Bloch-Torrey operator
\begin{equation}
B_g: =-\Delta_{x,y} + i g x\,,
\end{equation}
 in $\Omega = \mathbb R^2 \setminus \cup_{\gamma \in \mathbb Z^2} H_\gamma$ as defined in \eqref{defd1}-\eqref{defd2} and where  $g$ is a  non zero real parameter. \\
We introduce
\begin{equation*}
\mathcal H:= L^2(\Omega) \quad \mbox{and} \quad \mathcal V:=\{u \in H_0^1(\Omega)\,,\, |x|^\frac 12 u\in L^2(\Omega)\}
\end{equation*}
and we can now extend the operator initially defined on
$C_0^\infty(\Omega)$ to get a closed operator on $\mathcal H$ by the
following variant of the Friedrichs extension.
\begin{proposition}\label{proposition2.1}
The Dirichlet  realization $\mathcal B^D_g $ of  $B_g$ with domain
\begin{equation}
D(\mathcal B^D_g )=\{ u\in  \mathcal V\,,\, \mathcal B^D_g u \in \mathcal H\}
\end{equation}
is a closed accretive operator which generates a continuous semi-group
on $L^2(\Omega)$.
\end{proposition}
\begin{proof}
We take 
\begin{equation*}
\mathcal H:= L^2(\Omega) \mbox{ and } \mathcal V:=\{u \in H_0^1(\Omega)\,,\, |x|^\frac 12 u\in L^2(\Omega)\}
\end{equation*}
and apply Theorem \ref{LaxMilgramv2} to the quadratic form
\begin{equation*}
a(u,u)= \int_\Omega |\nabla u|^2\, dx \, dy + i \int_\Omega x |u|^2 \, dxdy  + C_0 \, ||u||^2\,,
\end{equation*}
where $||u||$ denotes the $L_2(\Omega)$ norm and $C_0>0$ is a large
positive constant to be determined.\\
For $\Phi_1=\Phi_2$ we take the multiplication operator by $x/\sqrt{1
+ x^2}$.\\
We simply then observe that
\begin{equation*}
\Re a (u,u) = \int_\Omega |\nabla u|^2\, dx \, dy + C_0 \, ||u||^2\,,
\end{equation*}
and then consider $ \Im a (u, \Phi_1(u))$.\\
By an easy computation, we can show the existence of $C>0$ such that,
for all $u\in \mathcal V$,
\begin{equation*}
\Im a (u, \Phi_1(u))\geq  \int_\Omega x^2 (1+x^2)^{-1/2}  |u|^2 \, dxdy - C\, ||u||^2_{H^1(\Omega)}
\end{equation*}
 Hence the
assumptions of the theorems are satisfied if we choose $C_0$
sufficiently large. What we call $\mathcal B^D$ is the operator $S -
C_0$ where $S$ is given by Theorem \ref{LaxMilgramv2}.
To get the semi-group property, we then apply the Hille-Yosida Theorem.
\end{proof}

\section{Floquet approach for $y$-periodic problems}
\label{sec:Floquet}

The goal is now to analyze the spectrum of this operator and  to
possibly consider asymptotic problems in function of $g$.
Starting with a qualitative analysis, we omit the reference to 
$g>0$.

\subsection{Floquet decomposition}

 Let $\tau_2$ be the translation by $(0,1)$, i.e.
\begin{equation*}
\tau_2 u(x,y) = u (x,y-1) \mbox{ for } u\in L^2(\Omega)\,.
\end{equation*}
Observing that
\begin{equation}
\tau_2 \circ \mathcal B^D =  \mathcal B^D \circ \tau_2\,,
\end{equation}
we can, at least formally, apply a Floquet theorem in the $y$ variable
and get the family of operators ($q \in \mathbb R$)
\begin{equation}
 B_{g,q}  : =- \left(\frac d {dy} -i q \right)^2 - \frac{d^2}{dx^2}+ i g x\,,
\end{equation}
in 
\begin{equation*} 
\Omega_0 = \biggl(\mathbb R \times (-1/2,1/2)\biggr) \setminus  \bigl\{\cup_{n \in \mathbb Z} H_{(n,0)}  \bigr\}\,,
\end{equation*}
where we put the Dirichlet condition at the boundary of each $ H_{(n,0)}$ and
the periodicity condition on the remaining boundary
\begin{equation}
u(x,- 1/2)=u(x, 1/2)\,,\quad
\partial_y u (x,- 1/2)= \partial_y u(x, 1/2)\,,\quad \forall x\in \mathbb R\,.
\end{equation}
Nevertheless the best way is to consider it as an operator on the
infinite perforated cylinder
\begin{equation*}
\widehat \Omega_0 := \bigl(\mathbb R \times \mathbb T^1\bigr) \setminus  \bigl\{\cup_{n \in \mathbb Z}  H_{(n,0)} \bigr\} \,,
\end{equation*}
where $\mathbb T^1=\mathbb R/\mathbb Z$ with Dirichlet condition on
$\partial \widehat \Omega_0$.

\subsection{Spectral analysis}

We now study the spectral properties of the Dirichlet realization
$\mathcal B^D_q $ of $B_{g,q}$. We omit from now on the reference to
$g>0$ which is unimportant for the qualitative analysis.
\begin{proposition}\label{Q1.2} 
For any $q\in \mathbb R$, we can extend $B_q$ as a closed operator
$\mathcal B^D_q $ with the following properties:
\begin{enumerate}
\item  $\mathcal B^D_q $ is a closed accretive
operator on $L^2(\Omega_0)$ and is the generator of a continuous
semi-group.
\item  
$\mathcal B^D_q  $ has compact resolvent.
\end{enumerate}
\end{proposition}
\begin{proof}
The proof is the same as for Proposition \ref{proposition2.1} but this
time we take
\begin{equation*}
\mathcal H_0:= L^2(\Omega_0) \mbox{ and } \mathcal V_0 
:=\{u \in H^1(\Omega_0)\,,\, |x|^\frac 12 u\in L^2(\Omega_0) \mbox{ and  (BC) holds} \}\,,
\end{equation*}
where by $(BC)$ we mean that $u$ satisfies the Dirichlet condition at
the boundary  of the $H_{(0,n)}$ and that we have the
periodicity condition $u(x,-\frac 12)=u(x,\frac 12)$ for $x\in \mathbb
R$.\\ Actually, it is better to deal with $\widehat \Omega_0$ as
mentioned above and to consider
\begin{equation*}
\widehat {\mathcal H}_0:= L^2(\widehat \Omega_0) \mbox{ and }\widehat { \mathcal V}_0 
:=\{u \in H_0^1(\widehat \Omega_0)\,,\, |x|^\frac 12 u\in L^2(\widehat \Omega_0)  \}\,.
\end{equation*}

The compact resolvent property is a consequence of the compact
injection of $\mathcal V_0$ in $\mathcal H_0$.\\

Noting the inequality
\begin{equation*}
\Re \langle \mathcal B^D_q  u\,,\, u\rangle \geq 0\,,\quad \forall u \in \mathcal V_0 \,,
\end{equation*}
we can construct by the Hille-Yosida theorem the associated continuous
semi-group and consider (as in \cite{MMGghm})
\begin{equation*}
\mathbb R_+ \ni t \mapsto \exp (- t \mathcal B^D_q) \,.
\end{equation*}
\end{proof}

\subsection{Application of Floquet decomposition to spectral theory}\label{ss3.3}

The interest of the Floquet theory is to relate the spectrum of the
initial operator to the union of the spectra of a family of simpler
operators. It is based on the so-called Floquet decomposition whose
definition is independent of the operator but is only related to the
symmetry group property satisfied by the operator.  The application to
spectral theory is standard in the self-adjoint case (see for example
in Reed-Simon \cite{ReSighm}) but we are not aware of a reference for
the application in the non-self-adjoint case (see nevertheless
\cite{Ku0ghm,Kughm}). Therefore it is not completely clear that the
proof easily goes on in the non-self-adjoint case and this is why we
present in this subsection an approach covering this situation.  Our
main result is:
\begin{proposition}\label{Prop3.2}
\begin{equation}
\overline{ \cup_{q  \in \mathbb R} \sigma(\mathcal B^D_q )} \subset  \sigma ( \mathcal B^D) \,,
\end{equation}
 where, for $A\subset \mathbb C$, $\overline{A}$ denotes the closure of $A$.
\end{proposition}
\begin{proof}~
\paragraph{Step 1} 
We first observe that $\mathcal B^D_q $ is  unitarily  equivalent to
$\mathcal B^D_{q + 2\pi}$. Hence it is enough to consider the problem
for $q \in [0,2\pi]$.

\paragraph{Step 2} 
This is a simple version of the Schnoll theorem relating the spectrum
and the existence of generalized polynomially bounded
eigenfunctions. Let $q \in \mathbb R$. We know from Proposition
\ref{Q1.2} that $\sigma(\mathcal B_q ^D)$ is either empty or consists
of eigenvalues with finite multiplicity. There is nothing to prove if
$\sigma(\mathcal B_q ^D)$ is empty. Let us assume now that it is not
empty and let $\mu$ be an eigenvalue of $\mathcal B^D_q $ associated
to a normalized eigenfunction $u_\mu (x,y)$. We extend $u_\mu$ as a
periodic function in the whole domain $\Omega$ and consider a cutoff
function $\chi$ which equals $1$ on $(-\frac 12,\frac 12)$ with
support in $(-1,+1)$.\\
We then consider 
\begin{equation*}
u_{\mu,n} (x,y) := e^{iy q } \chi (y/n) u_\mu (x,y)\,.
\end{equation*}
We have, for $n\in \mathbb N^*$, 
\begin{equation}\label{lb1}
||u_{\mu,n}||_{L^2(\Omega)}^2 \geq \int_{\Omega \cap\{ |y| <\frac n 2\}} |u_{\mu}|^2 dx dy   
=   n  \int_{\Omega \cap\{ |y| <\frac 1 2\}}  |u_{\mu}|^2 dx dy = n\,,
\end{equation}
and 
\begin{equation*}
f_n:=  (\mathcal B^D - \mu) u_{\mu,n} = - \frac 2n \chi'(y/n) \partial_y u_\mu (x,y)
  - \frac {1}{n^2} \chi''(y/n) u_\mu (x,y)\,.
\end{equation*}
It is immediate to see that there exists a constant $C$ such that
\begin{equation*}
|| f_n||_{L^2(\Omega)} \leq C/\sqrt{n}\,.
\end{equation*}
If $\mu$ was not in the spectrum of $\mathcal B_D$, we would have
\begin{equation*}
u_{\mu,n} = (\mathcal B_D-\mu)^{-1}  f_n = \mathcal O (1/\sqrt{n})\,,
\end{equation*}
in contradiction with \eqref{lb1}.
\end{proof}
  
We know that in the self-adjoint case (i.e. the case of the Schr\"odinger operator
with real periodic potential) the converse inclusion holds true. We notice in our
situation two conjectures. 
\begin{conjecture}\label{conjectureA}
\begin{equation}
\overline{ \cup_{q  \in \mathbb R} \sigma(\mathcal B^D_q )}= \sigma ( \mathcal B^D) \,.
\end{equation}
\end{conjecture}
We will see below that this conjecture is a consequence of the second
one:
\begin{conjecture}\label{conjectureB}
$\mathcal B^D$ has the following property:
\begin{equation}\label{conjdens}
\overline{\rho (\mathcal B^D)} = \mathbb C\,,
\end{equation}
where $\rho (\mathcal B^D)$ denotes the resolvent set of $\mathcal B^D$.
\end{conjecture}

\begin{remarks}
\begin{enumerate}~
\item
Property \eqref{conjdens} would have been automatically satisfied if
$\mathcal B^D$ was self-adjoint (because the spectrum would have been
real).
\item 
The proof of the second conjecture (which would imply the first one as
can be seen from the proof below) does not seem at the moment easier
as the first one.
\item 
The validity of Conjecture \ref{conjectureB} would probably be quite
specific of the Bloch-Torrey operator and of the case $d=2$.  For
instance, we suspect that the  resolvent set  of the operator
$-\Delta + \cos x + i \cos y$ in $\mathbb R^2$ has not this
property.\\  We indeed suspect that the spectrum can be described
by a Floquet theory as $\cup_{p,q} (\lambda_1(p)+\lambda_2(q))$, where
$\lambda_1(p) $ is the Floquet eigenvalue of $(-\frac{d^2}{dx^2} +
\cos x)$ and $\lambda_2(q)$ is the Floquet eigenvalue of
$(-\frac{d^2}{dy^2} +i \cos y)$ which is not real valued.
\item
It could be easier to show that the set
\begin{equation*}
\Omega^R:=   \mathbb C \setminus  (\overline{\cup_{q  \in [0,2\pi]} \sigma(\mathcal B^D_q )})\,,
\end{equation*}
satisfies 
\begin{equation*}
\overline{\Omega^R} = \mathbb C\,.
\end{equation*}
The set $\cup_{q \in \mathbb R} \sigma(\mathcal B^D_q )$ seems indeed
a union of piecewise analytic curves in $\mathbb C$. We do not know if
this property will lead to the proof of Conjecture \ref{conjectureA}.
Again this question would be specific of the case $d=2$.

\item 
One could imagine that (for $g$ large) some of these curves
$[0,2\pi]\ni q \mapsto \lambda_k (q)$ are simple closed curves in
$ \mathbb C$ implying the non-connectedness of $\Omega^R$. This
situation is of course excluded in the self-adjoint case.
\end{enumerate}
\end{remarks}

We now prove a weaker version of Conjecture \ref{conjectureA} for
which we need more definitions.  Observing that all our operators are
accretive, we immediately see that
\begin{equation} 
\mathbb C_-:= \{z\in \mathbb C\,,\, \Re z < 0\} \subset \rho ( \mathcal B^D)\cap \Omega^R\,.
\end{equation}
We now denote by $\omega_D$ the connected component in $ \Omega_D:=\rho(\mathcal
B^D)$ containing $ \mathbb C_-$ and by
$\omega_R$ the connected component in $\Omega^R$ containing $ \mathbb C_- $.\\
From Proposition \ref{Prop3.2} we know that
\begin{equation*}
\omega_D \subset \omega_R\,.
\end{equation*}
We then prove:
\begin{proposition}
$\omega_D =\omega_R$.
\end{proposition}
\begin{proof}
By Floquet theory, there is a unitary transform $\mathcal U$ such that
\begin{equation}\label{eq:ft1}
  \mathcal U^{-1} (\mathcal B^D-\lambda)^{-1} \mathcal U 
  = \int_0^{2\pi} ( \mathcal B_q ^{D} -\lambda) ^{-1} d q \,.
\end{equation}
It is clear that this formula holds if $\lambda \not\in \sigma
(\mathcal B^D)$ but we will show that, under assumption
\eqref{conjdens} this can be extended if the right hand side is well
defined.
    
Suppose indeed that $\lambda \not\in \cup_{q \in [0,2\pi]}
\sigma(\mathcal B^D_q )$. Hence $(\mathcal B^D_q -\lambda)^{-1}$
exists for any $q \in [0,2\pi]$ and we have to verify (to give a
meaning of the right hand side in \eqref{eq:ft1}) that we have a
uniform control with respect to $q $. For this we observe that
\begin{equation*}
(\mathcal B_q ^D - \lambda)\circ (\mathcal B_{q _0}^D - \lambda)^{-1} = 
I + 2i(q  - q _0) \partial_y (\mathcal B_{q _0}^D - \lambda)^{-1} + (q ^2 - q _0)^2(\mathcal B_{q _0}^D - \lambda)^{-1} \,.
\end{equation*}
For $|q -q _0|$ small enough this implies 
\begin{equation*}
(\mathcal B_q ^D - \lambda)^{-1} =  (\mathcal B_{q _0}^D - \lambda)^{-1} (I + 2i(q  - q _0) 
\partial_y (\mathcal B_{q _0}^D - \lambda)^{-1} + (q ^2 - q _0)^2(\mathcal B_{q _0}^D - \lambda)^{-1})^{-1}\,.
\end{equation*}
By compactness, we get the uniformity and define an holomorphic
$L^2$-valued function $\tilde R(\lambda)$ extending the resolvent
$R(\lambda)$ to the open set $ \Omega^R$.   Let us assume by
contradiction that $\omega_D$ is strictly included in $\omega_R$.  Let
us consider $\lambda_0 \in \omega_R \setminus \omega_D$. Using
Lemma~\ref{lemW} and the remark there, we can construct a sequence
$u_n$ of normalized functions such that $(\mathcal B^D - \lambda_0)
u_n$ tends to zero. Now we consider the holomorphic functions $
\omega_R \ni \lambda \mapsto v_n(\lambda)=\tilde R(\lambda) (\mathcal
B^D - \lambda) u_n$. On $\omega_D$, one has $v_n(\lambda)=u_n$,
therefore $v_n(\lambda)$ is constant and equal to $u_n$ over
$\omega_R$. In particular, $||v_n(\lambda)||=1$ for all integer $n$
and all $\lambda \in \omega_R$. However, one has
$||v_n(\lambda_0)||=||\tilde R(\lambda_0) (\mathcal B^D - \lambda_0)
u_n|| \to 0$, hence we get a contradiction.
\end{proof}
\begin{rem}
Note that the proposition implies that $\partial \omega_D =\partial
\omega_R$.  One could hope that $ \partial \omega_R= \partial \Omega^R$ which would correspond to the guess
that all the other components of $\Omega^R$ are bounded by spectral
curves (as mentioned above).
\end{rem}

\section{The Bloch-Torrey operator in a perforated cylinder continued. }
\label{sec:x-dir}

Here the analysis is  quite specific of the Bloch-Torrey operator.

\subsection{The pseudo-invariance in the $x$-variable}

Let $\tau_1$ be the translation by $1$ in the $x$ variable, which is 
defined for $u\in L^2_{loc}(\overline{\Omega_0})$ by
\begin{equation*}
 \tau_1 u (x,y)=u(x-1,y)\,.
\end{equation*}
Using this translation, one obtains that:
\begin{equation}
\mbox{  If } \lambda \in \sigma(\mathcal B^D_q ) \mbox { then } \lambda + i g \in
\sigma(\mathcal B^D_q )\,.
\end{equation}
This indeed results from the commutation relation
\begin{equation}\label{commrel1} 
\tau_1 \circ \mathcal B^D_q  = (\mathcal B^D_q  - i g) \circ  \tau_1\,.
\end{equation}
Here it is important to note  in the argument that  we have used that 
\begin{equation*}
\tau_1 \Omega_0=\Omega_0\,.
\end{equation*}
 
\begin{rem}\label{Remark4.1}
This commutation relation does not permit to use the standard Floquet
theory (see \cite{ReSighm,Ku0ghm,Kughm,Wighm}). We can try to recover
this invariance by considering on the Hilbert space $\oplus_{n\in
\mathbb Z} \mathcal H_n$ (with $\mathcal H_n$ isomorphic to
$L^2(\widehat \Omega_0)$) the operator $\oplus_{n\in \mathbb Z}
(\mathcal B^D_q + i gn )$. If we denote by $\tau_0$ the translation on
$\ell^2(\mathbb Z)$, we obtain the commutation of this operator with
respect to $\tau_0^{-1} \tau_1$. One can then perform a Floquet
decomposition in this context. This will give another way to obtain
our next results but not more.
\end{rem}
 
Using Proposition \ref{Q1.2} (second assertion), we get
\begin{proposition}
For any $q\in \mathbb R$, there is a discrete sequence (finite,
infinite or possibly empty) of eigenvalues $\lambda_k(q ) \in \mathbb
R_+ + i g (-\frac 12,\frac 12]$ such that the spectrum is given by
\begin{equation}\label{eq:4.2}
 \sigma ( \mathcal B^D_q ) = \cup_{k,n} (\lambda_k(q )  +i g n)\,.
\end{equation}
Moreover, if $u$ is an eigenfunction of $\mathcal B^D_q $ for
$\lambda_k(q ) $, then $\tau_1^n u$ is an eigenfunction for
\begin{equation*}
\lambda_{k,n} (q ):= \lambda_k(q )  +i g n \,.
\end{equation*}
\end{proposition}
 
Nevertheless it remains to analyze these sequences $\lambda_k(q )$ and
to determine in particular if the spectrum is not empty (a question
which is specific of the non-self-adjoint situation).\\ We now
consider the associate semi-group. We deduce from \eqref{commrel1}
\begin{equation}\label{commrel2} 
\tau_1\circ  \exp (- t  \mathcal B^D_q)  = \exp (it g) \,  \exp (- t  \mathcal B^D_q)  \circ  \tau_1\,.
\end{equation}
Choosing 
\begin{equation*}
t_g = \frac{2\pi}{g}\,,
\end{equation*}
we  get
\begin{equation}\label{commu}
\tau_1\circ  \exp (- t_g \mathcal B^D_q)  =  \,  \exp (- t_g  \mathcal B^D_q)  \circ  \tau_1\,.
\end{equation}
The idea is now to consider the spectrum of the operator
\begin{equation}
K_{g,q }:= \exp (- t_g  \mathcal B^D_q) \,,
\end{equation}
and to perform an adapted Floquet decomposition for this operator
since we have recovered the commutation relation.\\

A natural question is then to ask if $K_{g,q }$ is a compact operator
on $L^2(\widehat \Omega_0)$. This would be clear if $A:=\mathcal B^D_q
$ was self-adjoint or more generally sectorial.  In this case,
one can indeed prove that for any $T>0$, $t A \exp (- t A)$ is bounded
uniformly for $t\in (0,T)$ (see for example \cite{EN}, Chapter 2 and
references therein). Hence the operator  $\exp(-t A) = (A+\lambda_0)^{-1}
\left((A+\lambda_0) \exp(-t A)\right)$ is compact as the composition
of a compact operator and a bounded operator.  But $\mathcal B^D_q $
is not sectorial (at least if the spectrum is not empty) due to
\eqref{eq:4.2}.\\
On the other hand, the guess is that its spectrum is given by
\begin{conjecture}
\begin{equation}
\sigma ( K_{g,q })\setminus \{0\} = \cup_k \exp (- t_g \lambda_k(q ))\,.
\end{equation}
\end{conjecture}
 The first result, which may be non intuitive, is:
\begin{proposition}
If $K_{g,q }$ is a compact operator, then the spectrum of $\mathcal
B^D_q $ is empty.
\end{proposition}
\begin{proof}~
\paragraph{Step 1} 
If $\lambda_k$ is an eigenvalue of $\mathcal B^D_q $ and $u_k $ 
is the corresponding normalized eigenfunction, then $u_k$ is an
eigenfunction of $\exp (- t_g \mathcal B^D_q )$ associated with
$\lambda = \exp (- t_g
\lambda_k)$. \\
This proves at least that:
\begin{equation}\label{eq:4.7}
 \cup_k \exp (- t_g \lambda_k(q)) \subset \sigma ( K_{g,q })\setminus \{0\} \,.
\end{equation}
 
\paragraph{Step 2}
By \eqref{commu} $\tau_1 u_k$ is also an eigenfunction. More generally
$u_{k,n}:= \tau_1^n u_k$ is an eigenfunction for any $n\in \mathbb
Z$.\\
It remains to show that the vector space  $\mathcal U$ formed by the
$\tau_1^nu_k$ ($n\in \mathbb Z$) is not of finite dimension which will
give a contradiction to the compactness.\\
Suppose by contradiction that $\mathcal U$ is of finite
dimension $N_0$, then there exists $N_1$ such that this space is
generated by the $\tau_1^nu_k $ for $|n| \leq N_1$. Hence (uniform
localization of a normalized element belonging to a finite subspace in
$L^2(\Omega_0)$) for any $\epsilon >0$, there would exist $R>0$ such
that
\begin{equation*}
|| \tau_1^n u||_{L^2((-R,R) \times (-\frac 12,\frac 12) \cap \Omega_0)} \geq 1-\epsilon\,,\, \forall n \in \mathbb Z\,.
\end{equation*}
But this cannot be true as $|n| \rightarrow +\infty$, because, for
fixed $R$, $ || \tau_1^n u||_{L^2((-R,R) \times (-\frac 12,\frac 12)
\cap \Omega_0)}$ tends to $0$.
\end{proof}   
\begin{rem}
In the self-adjoint case, we could have directly used that the
eigenfunctions corresponding to different eigenvalues are orthogonal.
This is not the case here. In the case $q =0$, we can observe that if
$u$ is an eigenfunction of $\mathcal B_{0}^D$ associated with
$\lambda$, then $\bar u$ is an eigenfunction of the realization of
$-\Delta -ix $ for the eigenvalue $\bar \lambda$. Hence, we get
\begin{equation}
\langle u_{k,n}\,,\, \bar u_{k,n'} \rangle =0 \,,\, \mbox{ for } n\neq n'\,.
\end{equation}
\end{rem}

\begin{rem}
In the case without holes, the  emptiness of the spectrum for
the operators $-\frac{d^2}{dx^2}+i g \, x -\frac{d^2}{dy^2}$ in
$\mathbb R^2$ or $\mathbb R \times (-\frac 12,\frac 12)$ is known for
a long time (see \cite{Hghm} and references therein). Nevertheless our
guess is that this is not always the case for our perforated situation
(see for example \cite{AGHghm} for the exterior problem).
\end{rem}

\subsection{Floquet decomposition in the $x$-direction.}

The authors of \cite{MMGghm} analyze directly the spectrum of $K_{q
,g}$.  Observing as in \cite{MMGghm} that $K_{q ,g}$ commutes with the
translation $\tau_1$, we now consider a Floquet decomposition
(relative to the translation in $x$).  We recall from this
decomposition that, if $u$ is in $L^2(\Omega_0)$, then
\begin{equation}\label{fl1}
u_p :=  \sum_n e^{-inp} \tau_1^n u 
\end{equation}
satisfies the $p$-Floquet condition 
\begin{equation*}
\tau_1 u_p = e^{ip} u_p \,,
\end{equation*}
is well defined  in $L^2_{loc}$ and satisfies
\begin{equation}\label{isom}
 \int_{\Omega_0}  |u(x,y)|^2 dxdy  = \int _{(0,2\pi) } || u_p (x,y)||^2_{L^2(\Omega_0 \cap (-\frac 12,\frac 12)^2)}\,dp \,.
\end{equation}
Conversely, if we have a measurable family (with respect to $p$) of
functions $u_p$ in $L^2(\Omega_0)$ satisfying the $p$-Floquet
condition and such that $\int _{(0,2\pi) } || u_p
(x,y)||^2_{L^2(\Omega_0 \cap (-\frac 12,\frac 12)^2)}\,dp$, then
\begin{equation*}
  u = \frac{1}{2\pi} \int_{(0,2\pi)} u_p(x,y) \, dp\,,
\end{equation*}
is well defined in $L^2(\Omega_0)$ and satisfies \eqref{isom}.\\ 
Now, if $u$ is an eigenfunction associated with an eigenvalue $\lambda$
of $\mathcal B^D_q $, then $u_p$ satisfies
\begin{equation*}
 (\exp (- t_g \mathcal B_q^D)) u_p = \exp (- t_g \lambda) \,  u_p\,.
\end{equation*}
  
More generally, what one hopes from Floquet theory is that 
\begin{equation}
 \sigma (K_{q ,g}) \setminus \{0\} = \cup_{p\in \mathbb R}( \sigma (K_{q ,g,p}) \setminus \{0\})
\end{equation}
where $ K_{q ,g,p} $ is the restriction of $K_{q ,g}$ to the functions
satisfying the $p$-Floquet condition with respect to the translation
$\tau_1$.

The authors formally observe the following rather surprising property:
\begin{proposition}
Suppose that $u_p$ is a Floquet eigenfunction of $K_{g,q ,p}$ then,
for any $t >0$, the function
\begin{equation*}
u_{p,t}:= \exp (- t \mathcal B_q^D) \, u_p
\end{equation*}
satisfies the Floquet condition with parameter $ p-tg$ and
\begin{equation*} 
\exp (- t_g \mathcal B_q^D) \, u_{p,t} = \exp (- t_g \lambda) \,  u_{p,t}\,.
\end{equation*}
\end{proposition}

\begin{proof}
By \eqref{commrel2}, we have 
\begin{equation*}
 \tau_1 \underbrace{(\exp (- t \mathcal B_q^D)) u_p}_{u_{p,t}} = \exp(i tg)  \, \exp (- t \mathcal B_q^D) \, \tau_1 u_p = \exp (- i (p- tg)) \, u_{p,t}\,.
\end{equation*}
\end{proof}

\vskip 5mm

As an application, this justifies at least formally to look
numerically at the periodic problem associated with $\exp (- t_g
\mathcal B_q^D)$ and to recover the spectrum of $\mathcal B_q $ by
considering $-\log \mu/t_g + i g \mathbb Z$ (for $\mu \in \sigma
(K_{q,g,0}))$.
Note that at this stage we have only proved (see \eqref{eq:4.7})  that
\begin{equation}\label{incl1}
  \exp (-t_g \sigma (\mathcal B_q ^D)) \subset  \sigma ( K_{q,g})\, .
\end{equation}

\paragraph{From Floquet eigenfunctions of $K_{q ,g}$ to eigenfunctions of $\mathcal B_{q }^D$.}~\\
We get from formula \eqref{fl1} that for an eigenfunction $u$ of
$\mathcal B_{q }^D$
\begin{equation}\label{revformula}
  u =\frac{1}{2\pi} \int_{0}^{2\pi} u_p \, dp \,.
\end{equation}
This formula is standard due to the particular choice of the $u_p$
through formula \eqref{fl1}.  Note for example that 
\begin{equation*}
  \tau_1  u =\frac{1}{2\pi} \int_{0}^{2\pi} \tau_1 u_p \, dp =  \frac{1}{2\pi} \int_{0}^{2\pi}  e^{ip} u_p \, dp   \,.
\end{equation*}
Hence the infinite dimensional space $ {\rm Span} (\tau_1^n u) $ is
recovered by $\frac{1}{2\pi} \int_{0}^{2\pi} \beta_p u_p \, dp $ for
some function $\beta_p$.\\

The next proposition is a kind of converse statement.
\begin{proposition}\label{prop4.6}
If $\mu$ is an eigenvalue of $K_{q ,p=0}$ with corresponding
eigenfunction $u_0$, then $\log \mu + ig \mathbb Z$ belongs to the
spectrum of $\mathcal B_{q }^D$ and, for each $k$, we can construct
starting from $u_0$ an eigenfunction $u_{\lambda_k}$ associated with
$\lambda_k:= \log \mu + i g k$.
\end{proposition}
 \begin{proof}~\\
 
 {\bf Step 1: The heuristics behind the proof.}
 Heuristically, we will proceed in the following way. 
We start from an eigenfunction $u_0$ and associate with it
\begin{equation}\label{revformulab}
u_p =  \exp \biggl(- \frac{p}{g}( \mathcal B_q ^D-\lambda_0)\biggr) u_0\,.
\end{equation}
Defining $u$  by \eqref{revformula}, we obtain formally:
\begin{equation}\label{eq:revbb}
  \begin{array}{ll}
  ( \mathcal B_q^D -\lambda_0) u &= \bigl(\frac{1}{2\pi} \int_{0}^{2\pi} (\mathcal B_q^D-\lambda_0)  
 \exp (- \frac{p}{g}( \mathcal B_q ^D-\lambda_0))  \, dp \bigr) u_0\\
  & = \bigl(\frac{g}{2\pi } \int_{0}^{2\pi} ( - \frac{d}{dp}   \exp (- \frac{p}{g}( \mathcal B_q ^D-\lambda_0) ))  \, dp \bigr) u_0\\
  &= \frac{g}{2\pi }  \bigl(I -   \exp (- \frac{2\pi}{g}( \mathcal B_q ^D-\lambda_0) ) \bigr)  u_0\\
  &=0\,.
  \end{array}
\end{equation}
We finally observe that $u$ is not $0$, we have indeed formally
\begin{equation*}
  \Sigma_n \tau_1^n u =u_0\,.
\end{equation*}

{\bf Step 2:   Mathematical proof}~\\
In Step 1,  we have been very formal, omitting in
particular the questions relative to the domain for $u$.
We now give a rigorous proof.  Let us first state the following
proposition
\begin{proposition} \label{prop4.5}
For any $s\in \mathbb R$, 
\begin{equation*} 
 \exp (- t  \mathcal B_q^D) \,  L^{2,s}(\Omega_0) \subset  L^{2,s}(\Omega_0) \,,\quad \forall t \geq 0\,,
\end{equation*}
and there exists  $\omega_s$ such that 
\begin{equation}  \label{eq:ineq_auxil}
 ||  \exp (- t  \mathcal B_q ^D)  v ||_{L^{2,s}(\Omega_0)} \leq \exp (\omega_s t)  \, || v ||_{L^{2,s}(\Omega_0) } \,,\quad
\forall v \in   L^{2,s}(\Omega_0)\,,
\end{equation}
where $L^{2,s}$ is the weighted space
\begin{equation*}
L^{2,s}(\Omega_0) = \{ v \in L^2_{loc} (\Omega_0)\,,\, \langle x\rangle ^{s/2} v \in L^2(\Omega_0)\} \,,
\end{equation*}
with $\langle x\rangle = \sqrt{1+x^2}\,.$
\end{proposition}

We can now give a sense to \eqref{revformula}-\eqref{revformulab}. We
observe that $u_0$ belongs to $L^{2,-s}(\Omega_0)$ for $s>\frac 12$.
This implies by the proposition that $u$ is well defined in particular
in $L^{2,-s} (\Omega_0)$.\\
We now prove that $u$ is actually in $L^2$. For this we use the
Floquet theory. We observe that
\begin{equation*}
u_p :=  \exp \biggl(- \frac{p}{g}( ( \mathcal B_q ^D-\lambda_0))\biggr) u_0  \in L^{2,-s} \,,
\end{equation*}
and that $u_p$ satisfies the $p$-Floquet condition and the uniform
bound:
\begin{equation*}
|| u_p ||_{ L^2((0,1)^2\cap \Omega_0)} \leq \check C || u_0 ||_{ L^2((0,1)^2\cap \Omega_0)}\,,\quad \forall p\in [0,2\pi]\,,
\end{equation*}
 for some constant $\check C > 0$.
This implies, by the reverse Floquet decomposition formula (see around
\eqref{isom}), that $u$ is indeed in $L^2(\Omega_0)$.

We now compute $\exp (- t ( \mathcal B_q ^D-\lambda_0)) u$. It is
clear using the periodicity of $ \hat s \mapsto \exp \bigl(-\hat  s ( \mathcal B_q
^D-\lambda_0)\bigr) u_0$, that
\begin{equation*}
\exp (- t ( \mathcal B_q ^D-\lambda_0)) \, u = u ~~ \mbox{ in } L^{2,-s}\,,\quad \forall t \in [0, 2\pi/g]\,.
\end{equation*}
By semi-group theory, we immediately get that $u\in D(\mathcal B_q ^D)$
and that $ (\mathcal B_q ^D-\lambda_0)u=0\,$ as claimed in \eqref{eq:revbb}.  
This achieves the proof of Proposition \ref{prop4.6} modulo the proof of Proposition \ref{prop4.5}.
  
  \end{proof}
  
\paragraph{Proof of Proposition \ref{prop4.5}.}
To prove this proposition, we have to analyze the semi-group
\begin{equation*}
   t \mapsto \langle x\rangle^{s/2} \, \exp (- t \mathcal B_0^D)\,  \langle x\rangle^{-s/2}\,,
\end{equation*}
on $L^2(\Omega_0)$.\\
Its associated infinitesimal generator is the unbounded closed
operator
\begin{equation*}
\mathcal C_s^D:= \langle x\rangle^{s/2}\, \mathcal B_0^D  \, \langle x\rangle^{-s/2}
\end{equation*}
on $L^2(\Omega_0)\,.$ The associated differential operator is, with
$\alpha_s:= \langle x\rangle^{s/2}\,$, 
\begin{equation*}
   - \frac{d^2}{dx^2} -\frac{d^2}{dy^2} + ig x  -2 \alpha'_s\alpha_s^{-1}  \frac{d}{dx} - \alpha_s (\alpha_s^{-1})''\,,
\end{equation*}
which can be rewritten in the form
\begin{equation*}
   - \frac{d^2}{dx^2} -\frac{d^2}{dy^2}+ i g x  + a_s(x)  \frac{d}{dx} + b_s(x)\,,
\end{equation*}
where $a_s$ et $b_s$ are bounded.\\
The introduced perturbation does not change the form domain and using
Hille-Yosida theorem and the results of Appendix \ref{AppA}, we get
the existence of $\omega_s$ such that
\begin{equation*}  
   || \langle x\rangle^{s/2}\, \exp (- t \mathcal B_0^D) \,  \langle x\rangle^{-s/2} 
||_{\mathcal L(L^2 (\Omega_0))} \leq \exp (\omega_s t) \,,\quad \forall t>0\,.
\end{equation*}
This proves our proposition. \\

\begin{rem}
We guess but have not proved that $K_{q ,p=0}$ has only point
spectrum, i.e. that its spectrum only consists of eigenvalues. Note
that we could think of replacing in the above proof $u_0$ by a
sequence of approximate eigenfunctions $u_0^{(n)}$.\\
Formally \eqref{eq:revbb} gives
\begin{equation}\label{eq:revbba}
  ( \mathcal B_0^D -\lambda_0) u^{(n)} = \frac{g}{2\pi }  
\biggl(I - \exp \biggl(- \frac{2\pi}{g}( \mathcal B_q ^D-\lambda_0) \biggr)\biggr) \,  u_0^{(n)}\,.
\end{equation}
But on the right hand side we have only a sequence of periodic
functions which only tends to $0$ on a fundamental domain.\\ 
Hence, we do not have a Weyl sequence for $ \mathcal B_0^D$ relative
to $\lambda_0$. It remains some hope to proceed like in the proof of the
Schnoll theorem by introducing cut-off functions.
\end{rem}

\section{Quasi-modes and non-emptiness of the spectrum}
\label{sec:quasimodes}

We would like to analyze the behavior of the operator as $g\rightarrow
+\infty$ using the techniques of \cite{Henghm,AGH0ghm,AGHghm}.  We
take the semi-classical representation and look in the $2$-dimensional case at the asymptotic
limit as $h\rightarrow 0$ of the spectrum of the semi-classical
realization $\mathcal A _h^D$ of $-h^2 \Delta + i x $.\\

\subsection{Main results}  
The main result is:
\begin{thm}
\label{theorem5.1}   Under Assumptions \eqref{defd1}-\eqref{defd2} and assuming that $H_0$ is a strictly convex set in $\mathbb R^2$ with boundary of positive curvature, 
we have
\begin{equation}\label{limSpect1}  
  \lim\limits_{h\to0}\,\frac{1}{h^{2/3}}\inf \bigl\{\Re\, \sigma(\mathcal A _h^D) \bigr\} =  \frac{|a_1|}{2}\,,
\end{equation}  
where $a_1<0$ is the rightmost zero of the Airy function $\mathrm{Ai}$.\\

Moreover, for every $\varepsilon>0\,$, there exist $h_\varepsilon>0$
and $C_\varepsilon>0$ such that
\begin{equation}\label{estRes1}  
 \forall h\in(0,h_\varepsilon),~~~
 \sup_{
\begin{subarray}{c} 
\gamma\leq \frac{|a_1|}{2}  \\[0.5ex]
     \nu \in\mathbb{R}
\end{subarray}
}
 \|(\mathcal A _h^D-(\gamma -\varepsilon)h^{2/3}-i\nu)^{-1}\|\leq\frac{C_\varepsilon}{h^{2/3}}\,.
\end{equation}
In particular the spectrum of $\mathcal A_{h}^D$ is not empty.
\end{thm}

If one now looks at the reduced problem on the cylinder $\mathbb R
\times (-\frac 12,\frac 12)$, the main result would be the same for
$\mathcal A _h^{D,q }$ with the difference that we know that the
spectrum is discrete.

\begin{thm} \label{Theorem5.2}
 Under the same assumptions as in Theorem \ref{theorem5.1},  there exists $\lambda(h,q )$ such that 
\begin{equation*}
\lim_{h\rightarrow 0} ( \lambda(h,q ) -i r) h^{-\frac 23}= |a_1| e^{-i\frac \pi 3} 
\end{equation*}
and such that
\begin{equation*}
\lambda(h,q ) + i k \in  \sigma(\mathcal A _h^{D,q } ), \quad \forall k \in \mathbb Z\,.
\end{equation*}
\end{thm}

This result is essentially a reformulation of the results stated by
Almog in \cite{Almghm} and in the case of the exterior problem in
\cite{AGHghm}. 

\begin{rem}\label{remgen}
According to \cite{Henghm, AGH0ghm,AGHghm}, similar results can be formulated for
Neumann and Robin boundary conditions, by adapting the proofs from \cite{AGHghm}.
\end{rem}

\subsection{Proofs.}

\subsubsection{Lower bound}
The proof is identical to the exterior case considered in
\cite{AGHghm} (Subsection 2.2). The fact that there is an infinite number of holes
instead one hole does not change the proof.  The assumptions that $V(x,y)=x$ and the strict convexity of $H_0$ permits to verify all the assumptions appearing in this subsection.  We recall the main lines
with the simplifications that our potential $V$ is simply
$V(x,y)=x$.\\

By lower bound, we mean
\begin{equation}
  \label{eq:5agh}
 \varliminf\limits_{h\to0}\frac{1}{h^{2/3}}\inf \bigl\{\Re\, \sigma(\mathcal
 A _h^D) \bigr\} \geq  \frac{|a_1|}{2} \,.
\end{equation}

We keep the notation of \cite[Section~6]{AGH0ghm} and \cite{AGHghm}.
For some $1/3<\varrho<2/3$ and for every $h\in(0,h_0]$, we choose two
sets of indices $\mathcal{J}_{i}(h)\,$, $ \mathcal{J}_{\partial}(h)
\,$, and a set of points
\begin{subequations}   \label{eq:25agh}
\begin{equation}
\big\{a_j(h)\in \Omega : j\in \mathcal J_i(h)\big\}\cup\big\{b_k(h)\in\partial {\Omega} : k\in \mathcal J_\partial (h)\big\}\,,
\end{equation}
such that $B(a_j(h),h^\varrho)\subset\Omega$\,,
\begin{equation}
    \bar\Omega \subset\bigcup_{j\in \mathcal{J}_{i}(h)}B(a_j(h),h^{\varrho})~\cup\bigcup_{k\in \mathcal{J}_{\partial }(h)}B(b_k(h),h^{\varrho})\,,
\end{equation}
and such that the closed balls $\bar B(a_j(h),h^{\varrho}/2)\,$, $\bar
B(b_k(h),h^{\varrho}/2)$ are all disjoint. \\

Now we construct in $ \mathbb R^2$ two families of functions
\begin{equation}
 (\chi_{j,h})_{j\in \mathcal J_i(h)} \mbox{ and } (\zeta_{j,h})_{j\in \mathcal J_\partial (h)}\,,
\end{equation}
such that, for every $x\in\bar\Omega\,$,
\begin{equation}
\sum_{j\in \mathcal J_i(h)}\chi_{j,h}(x)^2+\sum_{k\in \mathcal J_\partial (h)}\zeta_{k,h}(x)^2=1\,,
\end{equation}
\end{subequations}
and such that 
\begin{itemize}
\item $\mbox{ Supp } \chi_{j,h}\subset B(a_j(h),h^{\varrho})$ for
$j\in \mathcal J_i(h)$,
\item
 $\mbox{ Supp }\zeta_{j,h}\subset
B(b_j(h),h^{\varrho})$ for $j\in \mathcal{J}_\partial(h)\,$,
\item
$\chi_{j,h}\equiv 1$ (respectively $\zeta_{j,h}\equiv1$) on $\bar
B(a_j(h),h^{\varrho}/2)$ (respectively $\bar
B(b_j(h),h^{\varrho}/2)$)\,.
\end{itemize}

We now define the approximate resolvent as in \cite{AGHghm}
\begin{equation}
\label{eq:9agh}
  \mathcal R_h(z) =\sum_{j\in \mathcal \mathcal{J}_i(h)}\chi_{j,h}(\mathcal A_{h}-z)^{-1}\chi_{j,h} 
+  \sum_{j\in \mathcal \mathcal{J}_\partial (h)}q _{j,h}  R_{j,h}(z) q _{j,h}\,,
\end{equation}
where $R_{j,h}(z)$ is given by \cite[Eq. (6.14)]{AGHghm}, and
$q _{j,h}={\mathbf 1}_\Omega\zeta_{j,h}$.\\
We write 
\begin{equation}
\label{eq:7agh}
\mathcal R_h(z) \circ (\mathcal A_h^D -z) = I + \mathcal E (h,z)\,,
\end{equation}
where
\begin{equation}
\label{eq:8agh}
\begin{array}{ll}
  \mathcal E(h,z)  & = -
  h^2[\Delta,\chi_{j,h}] (\mathcal A_{h} -z)^{-1}\chi_{j,h} \\ & \quad  +\sum_{j\in
    \mathcal \mathcal{J}_\partial (h)}(\mathcal A_h-z)q _{j,h}  R_{j,h}(z)q _{j,h}\,. 
    \end{array}
\end{equation}
The control can be achieved as in \cite{AGHghm}. We may thus conclude
that for any $\epsilon>0$ there exists $C_\epsilon>0$ such that for
sufficiently small $h$
\begin{displaymath}
  \sup_{\Re z\leq h^{2/3}(\frac{|a_1|}{2}-\epsilon)}\| \mathcal E(h,z) \|\leq C\, \bigl(h^{2-2\rho -\frac 23} + h^{2\rho -\frac 23}\bigr)\,. 
\end{displaymath}
Since for sufficiently small $h$, $(I+\mathcal E(h,z))$ becomes
invertible, we can now use \eqref{eq:7agh} to conclude that for any
$\epsilon>0$ there exist $C_\epsilon>0$ and $h_\epsilon >0$ such that
for any $h\in (0,h_\epsilon]$
\begin{displaymath}
   \sup_{\Re z\leq h^{2/3}(\frac{|a_1|}{2} -\epsilon)}\|(\mathcal A_h^D-z)^{-1}\|\leq
   \frac{C_\epsilon}{h^{2/3}} \,.
\end{displaymath}
This completes the proof of \eqref{eq:5agh}.

\subsubsection{The proof of upper bounds}

According to Proposition \ref{Prop3.2}, it is enough to prove Theorem
\ref{Theorem5.2}. Without loss of generality it is enough to present the
proof for $q =0$. Up to multiplication by $e^{i y q }$, the quasimode
is the same and we can take it localized near the point $(r,0)$ (see
Formula (7.2) in \cite{AGH0ghm}).  To prove that
\begin{equation*}
  \varlimsup\limits_{h\to0}\frac{1}{h^{2/3}}\inf \bigl\{\Re\, \sigma(\mathcal
 A _h^{D,q }) \bigr\} \leq  \frac{|a_1|}{2}\,.
\end{equation*}
we use the same procedure presented in \cite[Section 7]{AGH0ghm}.

\section{Conclusion}
\label{sec:conclusion}

The goal of the paper was to analyze mathematically the approach of
\cite{MMGghm} for the spectral analysis of the Bloch-Torrey equation in
periodic perforated domains in $\mathbb R^2$. 

We have proved that, at least in an asymptotic regime ($g\rightarrow
+\infty$), the presence of holes creates indeed spectrum and that it
was natural to detect these eigenvalues by analyzing an associate
bounded operator on $L^2(\mathbb T^2)$ through an extension of the
Floquet theory.\\

 We have focused on the proofs for the Dirichlet realization
$\mathcal B^D$ in planar domains. The precise definition of the
realizations for the other boundary conditions is treated in detail in
the $(1D)$-case in \cite{GHHghm} (see also \cite{GHghm} in (2D)) referring
to the same  variant of the Lax-Milgram theorem recalled in the appendix. The
semi-classical analysis of Section~\ref{sec:quasimodes}  relies
on Refs. \cite{Henghm,AHghm,AGH0ghm,AGHghm,GHghm}.  Note that Ref.
\cite{Henghm} is written for problems without the restriction to the
dimension. Nevertheless, as mentioned in the remarks of Subsection
\ref{ss3.3} some of the results and conjectures are specific of the
dimension~2.

\appendix

\section{Generalized Lax-Milgram Theorem}\label{AppA}

In this appendix we recall some results established in Almog-Helffer
\cite{AHghm}. For other applications of these results we also mention
\cite{KRRS} where magnetic Laplacians are considered.

 We consider two Hilbert spaces $\mathcal{V}$ and  $\mathcal{H}$
such that $\mathcal{V}\subset \mathcal{H}$, and that for some $C>0$
and any $u\in \mathcal{V}$, we have
\begin{equation}
  \label{eq:7} \|u\|_\mathcal{H} \leq C \|u\|_{\mathcal{V}}\,.
\end{equation}
Suppose further that
\begin{equation}\label{densi}
\mathcal{V} \mbox{ is dense in } \mathcal{H}\,.
\end{equation}
Consider a continuous sesquilinear
form $a$ defined on $\mathcal{V}\times \mathcal{V}$:
\begin{displaymath}
  (u,v)\mapsto a(u,v)\,.
\end{displaymath}
Let
\begin{equation}\label{lm3}
D(S) =\{u\in \mathcal{V}\;|\; v\mapsto a(u,v) \mbox{ is continuous on } \mathcal{V} \mbox{ in the
  norm of } \mathcal{H}\}\,,
\end{equation}
and  define the operator $S:D(S)\to\mathcal{H}$ by
\begin{equation}
\label{lm4}
a(u,v) = \langle Su\,,\,v \rangle_\mathcal{H}\,,\quad \forall u \in D(S) \mbox{ and } \forall v\in \mathcal{V}\,.
\end{equation}
We have the following theorem:
\begin{thm}
\label{LaxMilgramv2}
Let $a$ be a continuous sesquilinear form satisfying for some $\Phi_1,
\Phi_2 \in \mathcal L (\mathcal{V})$
\begin{equation}
\label{lm5n}
|a(u,u)| + |a(u, \Phi_1(u))| \geq \alpha\, \|u\|_\mathcal{V}^2\,,\quad \forall u\in \mathcal{V}\,.
\end{equation}
\begin{equation}
\label{lm5ne}
|a(u,u)| + |a( \Phi_2(u),u)| \geq \alpha\, \|u\|_\mathcal{V}^2\,,\quad \forall u\in \mathcal{V}\,.
\end{equation}
Assume further that $\Phi_1$ and $\Phi_2$ extend into continuous
linear maps in $\mathcal L (\mathcal{H})\,$ and let $S$ be defined by
\eqref{lm3}-\eqref{lm4}. Then
\begin{enumerate}
\item $S$ is bijective from $D(S)$ onto $\mathcal{H}$ and $S^{-1}\in \mathcal{L}(\mathcal{H})\,$;
\item $D(S)$ is dense in both $\mathcal{V}$ and $\mathcal{H}$;
\item $S$ is closed;
\item Let $b$ denote the conjugate sesquilinear form of $a$, i.e.
  \begin{equation}
\label{eq:53}
(u,v) \mapsto b(u,v):=\overline{a(v,u)}\,.
  \end{equation} 
Let $S_1$ denote the closed  linear operator associated with $b$ by
the same construction . Then  
\begin{equation}\label{infoadj}
  S^*=S_1 \mbox{ and }S_1^* = S\,.
  \end{equation}
\end{enumerate}
\end{thm}

\begin{rem}\label{remA3} 
We recall that the Hille-Yosida theorem (see Theorem 13.22 in
\cite{Hghm}) can be applied to the above defined operator $S$, if we
have
\begin{equation*}
\Re  \langle  S u, u\rangle \geq - C ||u||^2_{\mathcal H}\,,\quad \forall u \in \mathcal V \,.
\end{equation*}
\end{rem}

\section{Spectrum and Weyl's sequences}

In this section, we adapt to the non-self-adjoint situation a
characterization of the spectrum for a self-adjoint operator through
Weyl's sequences.
\begin{lem}\label{lemW}
Let $A$ be a closed operator in a Hilbert space $H$ which has the
property that, for any $\lambda \in \mathbb C$
\begin{equation}\label{hypsp}
  (A-\lambda) \mbox{ injective  implies } (A^*-\bar \lambda ) \mbox{ injective}.
\end{equation}
Under this assumption,  the two assertions are equivalent:
\begin{enumerate}
\item 
$\lambda \in \sigma (A)$\,.
\item 
There exists a sequence $u_n$ in $D(A)$ such that $||u_n||=1$ and
$(A-\lambda) u_n \rightarrow 0$.
\end{enumerate}
\end{lem}
\begin{proof}
The only difficulty is to prove that (1) implies (2). If there is no
sequence $(u_n) \in D(A)$ such that (2) is satisfied, then there
exists $c>0$ such that
\begin{equation*}
|| (A-\lambda) u|| \geq c ||u||\,,\quad \forall u \in D(A)\,.
\end{equation*}
From this we deduce that $(A-\lambda)$ is injective with closed
range. Now we have
\begin{equation*}
{\rm Range}  (A-\lambda) =( {\rm Ker} (A^* - \bar \lambda))^\perp\,.
\end{equation*}
 By \eqref{hypsp}  $(A^*-\bar \lambda)$ is injective. Hence we get the
surjectivity.
\end{proof}
\begin{remarks}\label{remW}~
\begin{itemize}
\item 
The property is evidently satisfied in the self-adjoint case because
the spectrum is real.
\item
The property \eqref{hypsp} is satisfied if $H$ is a complex
Hilbert space and if there is an antilinear involution $\Gamma$  such that $\Gamma D(A) \subset D(A)$ and 
\begin{equation}
 \Gamma A = A^* \Gamma\,.
\end{equation}
In particular it holds for our Bloch-Torrey operators by taking as
$\Gamma$ the complex conjugation.

\item 
The lemma is not true without property \eqref{hypsp}. As
suggested by the referee, one can indeed consider $H=L^2(0,1)$, $A=-i
\frac{d}{dx}$, $D(A)=H_0^1(0,1)$. We have $\sigma(A)=\mathbb C$ and
there is no a Weyl sequence for any $\lambda \in \mathbb C$. On the
other hand $(A-\lambda)$ is injective for any $\lambda$ and $(A^*-\bar
\lambda)$ is non-injective for any $\lambda$.
\end{itemize}
\end{remarks}

\begin{acknowledgments}
D.~S.~G. acknowledges a partial financial support from the Alexander
von Humboldt Foundation through a Bessel Research Award.  The
authors thank the anonymous referee for his remarks and suggestions,
in particular when observing that the statement of Lemma
\ref{lemW} that we initially gave in the submitted version, was
wrong.
\end{acknowledgments}

\end{document}